\documentclass[conference]{IEEEtran}
\IEEEoverridecommandlockouts
\usepackage{cite}
\usepackage{amsmath,amssymb,amsfonts}
\usepackage{algorithmic}
\usepackage{amsmath, amsthm, amsfonts}
\usepackage{tikz}
\usepackage{graphicx}
\usepackage{textcomp}
\usepackage{xcolor}
\usepackage{braket}
\usepackage{esvect}
\usepackage{balance}
\usepackage{bm}
\usepackage{adjustbox}
\makeatletter
\newcommand*{\rom}[1]{\expandafter\@slowromancap\romannumeral #1@}
\makeatother
\usepackage{mathtools}
\DeclarePairedDelimiter{\ceil}{\lceil}{\rceil}
\def\BibTeX{{\rm B\kern-.05em{\sc i\kern-.025em b}\kern-.08em
    T\kern-.1667em\lower.7ex\hbox{E}\kern-.125emX}}
\usepackage{mdframed}
\usepackage{lipsum}
\newtheorem{theorem}{Theorem}
\newtheorem{corollary}{Corollary}
\newtheorem{proposition}{Proposition}
\theoremstyle{definition}
\newtheorem{definition}{Definition}

\usepackage{mathtools}
\newcommand\numeq[1]%
  {\stackrel{\scriptscriptstyle(\mkern-1.5mu#1\mkern-1.5mu)}{=}}
\newcommand\numleq[1]%
  {\stackrel{\scriptscriptstyle(\mkern-1.5mu#1\mkern-1.5mu)}{\leq}}
\newcommand\numgeq[1]%
  {\stackrel{\scriptscriptstyle(\mkern-1.5mu#1\mkern-1.5mu)}{\geq}}
\DeclareMathOperator*{\argmax}{arg\,max}

\usepackage{environ}         
\usepackage{etoolbox}        
\usepackage{graphicx}        
\usepackage{pgfplots}

\DeclareRobustCommand{\rchi}{{\mathpalette\irchi\relax}}
\newcommand{\irchi}[2]{\raisebox{\depth}{$#1\chi$}}
\newcommand*{\QEDB}{\null\nobreak\hfill\ensuremath{\square}}%
\pgfplotsset{compat=1.17}
\begin{document}

\title{The Classical Capacity of Quantum Jackson Networks with Waiting Time-Dependent Erasures\\
}

\author{\IEEEauthorblockN{Jaswanthi Mandalapu, Krishna Jagannathan}
\IEEEauthorblockA{{Department of Electrical Engineering, IIT Madras} \\
ee19d700@smail.iitm.ac.in, krishnaj@ee.iitm.ac.in}
}

\maketitle
\begin{abstract}
We study the fundamental limits of classical communication using quantum states that decohere as they traverse through a network of queues. We consider a network of Markovian queues, known as a Jackson network, with a single source or multiple sources and a single destination. Qubits are communicated through this network with inevitable buffering at intermediate nodes. We model each node as a `queue-channel,' wherein as the qubits wait in buffer, they continue to interact with the environment and suffer a waiting time-dependent noise. Focusing on erasures, we first obtain explicit classical capacity expressions for simple topologies such as tandem queue-channel and parallel queue-channel. Using these as building blocks, we characterize the classical capacity of a general quantum Jackson network with waiting time-dependent erasures. Throughout, we study two types of quantum networks, namely, (i) Repeater-assisted and (ii) Repeater-less. We also obtain optimal pumping rates and routing probabilities to maximize capacity in simple topologies. More broadly, our work quantifies the impact of delay-induced decoherence on the fundamental limits of classical communication over quantum networks. 

\end{abstract}
\section{Introduction}
Quantum Internet, a global network interconnecting remote quantum devices through quantum links, is envisioned as a key aspect of the `second quantum revolution'~\cite{jk}. Quantum networks are known to enhance the capabilities of classical networks by executing protocols that are impossible to perform classically\cite{qn1,qn2,qn3,qn4,qn5,qn6}. 
A key challenge in implementing reliable end-to-end communication of quantum bits (or qubits) over quantum networks is that, unlike classical bits, qubits tend to \emph{decohere} rapidly \cite{Neilsenchuang}. 
The decoherence of a quantum state is due to its interaction with the environment, which leads to partial or complete loss of information. 



In this paper, we consider a setting where classical information is transmitted over a quantum network using qubits, and these qubits decohere as they traverse each node in the network. Specifically, we consider a network of Markovian queues with a single source or multiple sources and a single destination (Fig.~\ref{JacksonNetwork}). Qubits enter the network according to a Poisson process of fixed arrival rate. The network consists of a finite number of intermediate nodes interconnected with i.i.d. Bernoulli routing. We model each intermediate node as a $\cdot$/M/1 queue, i.e., the service time of each qubit at a given node is exponentially distributed with fixed service rate. In queuing literature, such a network is known as a Jackson network\cite[Sec.~7.7]{dsp}.

We adopt the \emph{queue-channel} framework, studied in \cite{qubitspaper,prabhaspawc,jsait}, to model waiting time-dependent erasures at each intermediate node. An erasure queue-channel is a non-stationary erasure channel with memory, where erasures occur due to the induced waiting times of qubits in a queue before processing\cite{prabhaspawc}. We refer to a network of $\cdot$/M/1 queue-channels as a \emph{Quantum Jackson Network.} Focusing on erasures, we characterize the information capacity of such a network from each source node to the destination node for fixed values of arrival rates and routing probabilities.  
 Quantum Jackson networks could find applications in repeater-assisted quantum communication\cite{p2}, multi-core quantum computers, and futuristic quantum data networks\cite{qcn}.


Throughout the paper, we consider two quantum network settings, namely \emph{repeater-assisted} networks and \emph{repeater-less} networks. Quantum repeaters are designed to improve the reliability of quantum communications by enhancing the fidelity of a quantum state. Motivated by this, in a repeater-assisted setting, we assume the coherence time of a qubit is effectively `restarted' at each intermediate node if the qubit is not already erased. On the other hand, in a repeater-less setting, we assume the coherence time is not affected at each node, i.e., the erasure probability of a particular qubit is a function of its total time spent in the network, which is the sum of the waiting times at each intermediate node traversed by the qubit before reaching the destination node. For the special case when the coherence times are exponentially distributed, the two settings turn out to be mathematically identical. 

For simplicity of analysis, we model the information loss solely due to the inevitable buffering times at each node in the network. We do not explicitly model the propagation delay and path loss, although it is possible to incorporate these aspects into our modeling framework. Likewise, we do not study hybrid networks consisting of repeater-assisted and repeater-less nodes. 


\subsection{Related Work and Contributions}
A series of recent papers \cite{qubitspaper,prabhaspawc,jsait} consider point-to-point quantum queue-channels and derive single-letter classical capacity for specific noise models, including erasures. Further, \cite{jsait} shows that the upper bound technique extends to the broader class of \emph{additive} queue channels. In work with a similar flavor, \cite{AC} investigates the channels with queue length dependent service quality, with motivation drawn from crowd-sourcing. Additionally, \cite{p2,p3,p4} studied the capacities of quantum networks with basic quantum channels, namely, bosonic channels, quantum-limited amplifiers, dephasing, and erasure channels in both repeater-assisted and repeater-less settings. However, the fundamental limits on classical communication when the qubits are transmitted over a network of queue-channels have not yet been understood.

In this paper, we investigate the fundamental limits of classical information transmission over a network of queue-channels. Our key contribution lies in using the queue-channel framework as a building block to study the classical capacity of a quantum Jackson network. First, using the tools like conditional independence lemma of a quantum queue-channel \cite[Lemma~1]{prabhaspawc}, and additivity result of Holevo information for quantum erasure channels \cite{Holevo}, we characterize the classical capacity expressions for simple topologies, namely, tandem queue-channel and parallel queue-channels for both network settings. Further, when coherence times are exponentially distributed, we derive the optimal pumping rates and routing probabilities to maximize the capacities in simple topologies. Finally, we generalize and obtain the information capacity from each source to the destination node for a general quantum Jackson network. We believe this is the first work to consider non-stationary erasure channels with memory and quantify the impact of delay-induced decoherence on the fundamental limits of classical communication over quantum networks.

\section{Preliminaries $\&$ System Model}\label{sec2}
\subsection{Quantum Erasure Queue-channel}
In this section, we revisit the framework of a point-point quantum erasure queue-channel introduced in \cite{prabhaspawc}. In a quantum queue-channel, qubits are processed sequentially over a single server queue in First Come First Serve (FCFS) fashion. Qubits enter the queue according to a stochastic process of fixed arrival rate. Each qubit is then served with an independent and exponentially distributed service time of fixed service rate. For stability in the queue, we assume the arrival rate is always less than the service rate. 

Let ${\rho_j} $ denote the density operator corresponding to the qubit state $j$. Each qubit takes a non-zero processing time to get served in the queue. Let $W_j$ denote the total sojourn time spent by the $j^{th}$ qubit in the queue. In order to capture the effect of decoherence, the erasure probability of each qubit is modeled as a function of its overall sojourn time in the queue. Specifically, the probability of erasure of a particular qubit $j$ is modeled as $p(W_j)$, where $p : [0,\infty) \to [0,1]$ is typically an increasing function in waiting time $W_j$. Let $E: \rho_j \to \ket{e}\bra{e}$ be an erasure operator that maps $j^{th}$ qubit to a fixed erasure state with probability $p(W_j)$. Then, a quantum erasure queue-channel parameterized by waiting time $W_j$ is represented as a map ${\mathcal{E}_{W_j}:S(H^I) \to S(H^O)}$ from the set of input Hilbert space to the set of output Hilbert space respectively. More formally, a quantum erasure queue-channel on qubit $\rho_j$ acts as follows: $${\mathcal{E}_{W_j}(\rho_j) = p(W_j) E \rho_{j} E^{\dagger} + q(W_j) \rho_{j}},$$ where ${q(W_j) = 1 - p(W_j)}$ is the probability of qubit being unaffected. Consequently, let $\mathbf{W} = (W_1, W_2, \ldots, W_n)$ be an $n-$length sequence of waiting times of qubits in the queue. Then, an $n-$fold quantum erasure queue-channel parameterized by the sequence of waiting times $\mathbf{W}$, is represented as a map ${\mathcal{E}_{\mathbf{W}}^{(n)} : S((H^I)^{\otimes n}) \to S((H^O)^{\otimes n})}$ from the set of all input Hilbert spaces to set of all output Hilbert spaces respectively. 

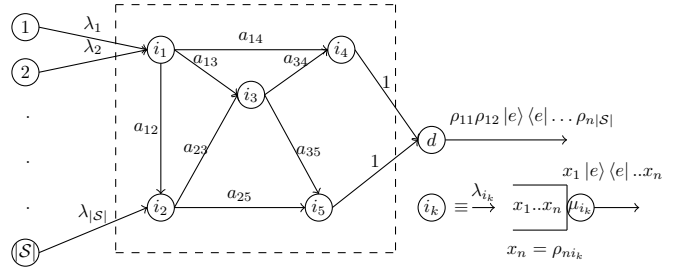
\begin{figure}[t]
\centering
    \begin{tikzpicture}[scale = 0.6,every node/.style={scale=0.7}]
    \draw (0,0) circle (0.3) node {$1$};
    \draw (0,-1) circle (0.3) node {$2$};
    \node[text width=0cm] at (0,-2) {$.$}; 
    \node[text width=0cm] at (0,-3) {$.$};
    \node[text width=0cm] at (0,-4) {$.$}; 
    \draw (0,-5) circle (0.3) node {${|\mathcal{S}|}$};
    \draw (3,-0.5) circle (0.3) node {$i_1$};
    \draw (3,-4) circle (0.3) node {$i_2$};
    \draw (5,-1.5) circle (0.3) node {$i_3$};
    \draw (7,-0.5) circle (0.3) node {$i_4$};
    \draw (6.5,-4) circle (0.3) node {$i_5$};
    \draw (9,-2.5) circle (0.3) node {$d$};
    \draw[->] (0.3,0) -- (2.7,-0.5) node[above,midway] {$\lambda_1$};
    \draw[->] (0.3,-1) -- (2.7,-0.5) node[above,midway] {$\lambda_2$};
    \draw[->] (0.3,-5) -- (2.7,-4) node[above,midway] {$\lambda_{|\mathcal{S}|}$};
    \draw[->] (3.3,-0.5) -- (4.7,-1.5) node[above,midway] {$a_{13}$};
    \draw[->] (3,-0.8) -- (3,-3.7);
    \draw[->] (5.3,-1.5) -- (6.7,-0.5) node[above,midway] {$a_{34}$};
    \draw[->] (5.3,-1.5) -- (6.5,-3.7);
     \node[text width = 0cm] at (6,-2.75) {$a_{35}$};
    \draw[->] (3.3,-4) -- (6.2,-4) node[above,midway] {$a_{25}$};
    \draw[->] (3.3,-4) -- (4.7,-1.5);
    \draw[->] (7.3,-0.5) -- (8.7,-2.5) node[above,midway] {$1$};
    \draw[->] (6.8,-4) -- (8.7,-2.5) node[above,midway] {$1$};
    \draw[->] (3.3,-0.5) -- (6.7,-0.5) node[above,midway] {$a_{14}$};
    \draw[->] (9.3,-2.5) -- (12,-2.5);
    \draw[dashed] (2,0.5) -- (2,-5);
    \draw[dashed] (2,0.5) -- (8.2,0.5);
    \draw[dashed] (2,-5) -- (8.2,-5);
    \draw[dashed] (8.2,0.5) -- (8.2,-5);
    \node[text width = 0cm] at (3.5,-2.7) {$a_{23}$};
    \node[text width = 0cm] at (2.4,-2.25) {$a_{12}$};
    \node[text width = 0cm] at (9.4,-2) {$\rho_{11}\rho_{12} \ket{e}\bra{e} \ldots \rho_{n |\mathcal{S}|}$};
    \draw (9,-4) circle (0.3) node {$i_k$};
    \node[text width = 0cm] at (9.5,-4) {$\equiv$};
    \draw (12,-3.5) -- (12,-4.5);
    \draw (10.8,-3.5) -- (12,-3.5);
    \draw (10.8,-4.5) -- (12,-4.5);
    \draw[->] (9.9,-4) -- (10.4,-4) node[midway,above] {$\lambda_{i_k}$};
    \draw (12.3,-4) circle (0.3) node {$\mu_{i_k}$};
    \node[text width = 0cm] at (10.8,-4) {$x_1 .. x_n$};
    \draw[->] (12.6,-4) -- (13.6,-4);
    \node[text width = 0cm] at (11.9,-3.2) {$x_1 \ket{e}\bra{e} .. x_n$};
    \node[text width = 4cm] at (13,-5) {$x_{n} = \rho_{ni_k}$};
    \end{tikzpicture}
\caption{An example of a Quantum Jackson Network with each intermediate node representing a quantum erasure queue-channel.}
\label{JacksonNetwork}
\end{figure}

\begin{theorem}\cite{prabhaspawc}  The classical capacity of a quantum erasure queue-channel is given by $\lambda \mathbb{E}_{\pi}[1-p(W)]$ bits/sec, where $\pi$ is the stationary distribution of the total sojourn time W.
\end{theorem}
Note that an $n-$fold quantum erasure queue-channel is neither a stationary nor a memoryless channel since the probability of $j^{th}$ qubit getting erased is a function of its waiting time which in turn depends on the waiting time of $(j-1)^{th}$ qubit and so on; further, entanglement is not necessary to achieve the capacity; see \cite{prabhaspawc,jsait} for more details. This work considers such \emph{non-stationary} quantum erasure queue-channels with \emph{memory} in a quantum network and characterizes the classical capacity.

Throughout this paper, we use the bold letter representation $\mathbf{X},\mathbf{W}$ to represent an $n-$length sequence.

\subsection{System Model}\label{jacksonnetworksetting}
Our system model considers a network of Markovian queues with a single source or multiple sources and a single destination; see Fig.~\ref{JacksonNetwork}. 
Each source node `$s$' independently generates a classical bitstream $\mathbf{X}^{(s)}$ over a finite input alphabet set $\rchi^n$, which is encoded into a sequence of possibly entangled qubit states $\boldsymbol{\rho}_{\mathbf{X}^{(s)}}$. Each source node transmits the qubits over the network using photons of a unique wavelength that is known at the destination. From each source $s$, the qubits enter the network according to a Poisson process of fixed transmission rate $\lambda_s$.

The network consists of a finite number of intermediate nodes interconnected with i.i.d. Bernoulli routing probabilities, given by a routing matrix $A = [a_{ij}],$ where $a_{ij}$ denotes the routing probability from node $i$ to node $j$. Let $\mathcal{I}$ denote the set of intermediate nodes in the network. We model each intermediate node $i \in \mathcal{I}$ as an $\cdot$/M($\mu_i$)/1 queue, i.e., each node $i$ serves the qubits with independent and exponentially distributed service times of rate $\mu_i$. In particular, we assume each intermediate node is a quantum queue-channel, where erasures occur due to the waiting times of the qubits in the queue. At destination $d$, we perform a general quantum measurement to decode the erased output sequence ${\mathbf{Y}^{(s)} \in \{\rchi \cup e\}^n}$. We assume the destination can recognize the qubits from each source node perfectly.
We refer to this network as a \emph{Quantum Jackson Network,} and derive the classical capacity (in bits/sec) from each source to the destination node.


\subsection{Types of Quantum Networks}\label{networktypes}
Throughout our work, we consider  two types of quantum Jackson networks as defined below.
\subsubsection{Repeater-assisted network} In this type of Jackson network, we assume each intermediate node is a quantum repeater. In other words, we assume the coherence times of qubits are statistically `restarted' at each intermediate node. 
\subsubsection{Repeater-less network} In this Jackson network, we assume the coherence times of the qubits are not affected at an intermediate node. Specifically, we model the erasure probability of a qubit as a function of its \emph{total} time spent in the network.

The above two types of quantum networks are studied extensively in the literature for the cases of bosonic channels, dephasing and memoryless quantum erasure channels; see \cite{p2},\cite{p3},\cite{p4}. In our work, we study the fundamental limits of communication over a quantum Jackson network where the qubits suffer  queuing delay-induced erasures.


\section{Building Blocks: Tandem and Parallel Queue-channels}\label{sec3}
In this section, we first characterize the classical capacities of two simple Jackson network topologies, namely tandem queue-channel and parallel queue-channels. Further, we derive the optimal transmission rate and optimal routing probabilities to maximize the capacities in simple topologies.
Using our results for these simple topologies as building blocks, we derive the classical capacity expression for any general Jackson network model in the next section.

\subsection{Tandem Queue-channel}\label{secrepeaterassisted}
We consider a network consisting of a single source node that transmits qubits over a network of $m$ intermediate nodes connected in tandem. Assume that the qubits enter the network according to a Poisson process of transmission rate $\lambda$. This network, depicted in Fig.~\ref{nqueuetandemnetwork}, is referred to as a tandem queue-channel. 

We denote $\textbf{W}_i = (W_{i1}, W_{i2}, \ldots, W_{in})$ as a sequence of waiting times of qubits at each node $i$. Let $\pi_i$ be the stationary distribution of the waiting times of the qubits at node $i$. 


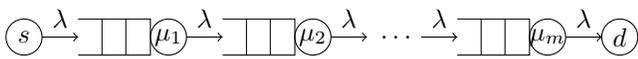
\begin{figure}[h!]
    \centering
\begin{tikzpicture}[scale=0.48]
\draw (-0.5,-0.5) circle (0.5) node {$s$};
\draw[->] (0, -0.5) -- (1,-0.5) node[above, midway] {$\lambda$};
\draw (1,0) -- (3,0) -- (3,-1) -- (1,-1);
\draw (1.66,0) -- (1.66,-1);
\draw (2.32,0) -- (2.32,-1);
\draw (3.5,-0.5) circle (0.5) node {$\mu_1$};
\draw[->] (4,-0.5) -- (5,-0.5) node[above, midway] {$\lambda$};
\draw (5,0) -- (7,0) -- (7,-1) -- (5,-1);
\draw (5.66,0) -- (5.66,-1);
\draw (6.32,0) -- (6.32,-1);
\draw (7.5,-0.5) circle (0.5) node {$\mu_2$};
\draw[->] (8,-0.5) -- (9,-0.5) node[above, midway] {$\lambda$};
\node[text width=3cm] at (12.5,-0.5) {$\ldots$}; 
\draw[->] (10.5,-0.5) -- (11.5,-0.5) node[above, midway] {$\lambda$};
\draw (11.5,0) -- (13.5,0) -- (13.5,-1) -- (11.5,-1);
\draw (12.16,0) -- (12.16,-1);
\draw (12.82,0) -- (12.82,-1);
\draw (14,-0.5) circle (0.5) node {$\mu_m$};
\draw[->] (14.5,-0.5) -- (15.5,-0.5) node[above, midway] {$\lambda$};
\draw (16,-0.5) circle (0.5) node {$d$};
\end{tikzpicture}
    \caption{Tandem Queue Erasure Network}
    \label{nqueuetandemnetwork}
\end{figure}
We present the following results stating the classical capacity of a tandem queue-channel under both network settings defined in Sec.~\ref{networktypes}. Consequently, we see that when the coherence times are exponentially distributed, the two types of tandem queue-channels are mathematically equivalent. 
\begin{theorem}\label{thm1}
The classical capacity of a repeater assisted tandem queue-channel (in bits/sec) is given by
$$
   C_{rt}(\lambda) = \lambda \mathbb{E}_{\vec{\pi}}\left[ \textstyle \Pi_{i=1}^m q(W_i)\right],
$$
where ${\vec{\pi} = (\pi_1,\pi_2, \ldots, \pi_m)}$ is a vector of stationary distributions of sojourn times in each queue respectively. 
\end{theorem}
\begin{proof}
We first prove the converse part of this theorem using Holevo information and the additivity result of Holevo for Quantum erasure channel, similar to \cite[Theorem~1]{prabhaspawc}. Next, the achievability is proved by fixing an encoding strategy, which considers independent and orthogonal quantum states at the transmitter. See Appendix~\rom{6}-A for the detailed proofs.
\end{proof}

Next, we provide the classical capacity of a repeater-less tandem queue-channel. Recall, in a repeater-less tandem queue-channel, the erasure probability of a particular qubit is modeled as a function of its overall sojourn time in the network. The following theorem states its classical capacity.


\begin{theorem}\label{thm3}
The classical capacity of a repeater-less tandem queue-channel (in bits/sec) is given by
$
   {C_{lt}(\lambda) = \lambda\left(1 - \mathbb{E}_{\vec{\pi}}\left[ p(\textstyle \sum_{i=1}^m W_i)\right] \right)}. 
$
\end{theorem}
\begin{proof}
The proof of this theorem is a direct consequence of the classical capacity of a queue-channel derived in \cite{prabhaspawc}. Note that, for any ergodic and stationary distribution of sojourn time $W$, the capacity of the queue-channel is proved to be ${\lambda (1-\mathbb{E}_{\pi}[p(W)])}$ bits/sec in \cite[Theorem~1]{prabhaspawc}. Further, a repeater-less tandem queue-channel can be seen as a single queue-channel with overall sojourn time ${W_1 + W_2 + \ldots + W_m}$. Hence, using the capacity expression of a quantum queue-channel, our result follows.
\end{proof}

\textit{Remarks:} The coherence time of a qubit is often modeled as an exponential random variable. That is, the probability of erasure $p(W)$ is modeled as ${p(W) = 1 - e^{-\kappa W}}$, where $1/\kappa$ is the characteristic time constant of the physical system under consideration. Moreover, from the properties of the Jackson network, the stationary waiting times in each queue of a tandem-queue channel are independent \cite{frank}. Accordingly, when the coherence times are exponential, the capacity expressions of the tandem queue-channel can be written as the product of the Laplace transforms of waiting times at each queue in the network. The following corollary specifies this result.

\begin{corollary}\label{corollary2}
The capacity of a tandem queue-channel when $p(W) = 1 - e^{-\kappa W}$ is given by
${
    C_t(\lambda) = \lambda \left(\textstyle \prod_{i=1}^m \frac{\mu_i - \lambda}{\kappa + \mu_i - \lambda}\right)
\textrm{bits/sec}}$.
\end{corollary}
\begin{proof}
Refer Appendix~\rom{6}-B.
\end{proof}

Using the above capacity expression, we now provide an optimal transmission rate of a homogeneous tandem queue-channel with $\mu= 1$; refer Fig.~\ref{tandemqcap} for the plot on capacity. We observe that there is an optimal transmission rate,  although the capacity curve is non-concave. Further, we see that with an increase in the value of $\kappa$, the system's capacity reduces, which happens due to the increase in delay-induced decoherence of qubits in the network.

The following proposition gives us a mathematical expression to calculate the optimal transmission rate in a homogeneous tandem queue-channel.
\begin{proposition}
The optimal transmission rate $\lambda^*$ that maximizes the capacity of a homogeneous tandem queue-channel with service rate $\mu=1$ is given by
${
    \lambda^* = \frac{1}{2}[2 + (m+1)\kappa - \sqrt{\kappa} \sqrt{4m\kappa + ((m+1)\kappa)^2}].
}$
\end{proposition}
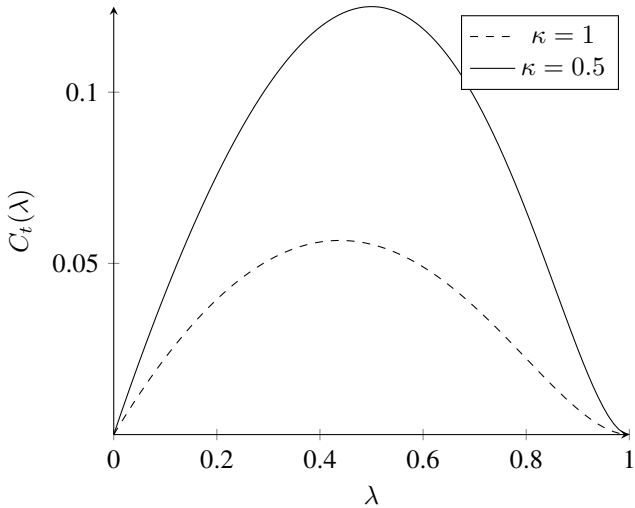
\begin{figure}
    \centering
    \begin{tikzpicture}[scale = 1]
        \begin{axis}[
        axis lines = left,
        xlabel = \(\lambda\),
        ylabel = {\(C_t(\lambda) \)},
       xtick={0,0.2,0.4,0.6,0.8,1},
       xticklabels={0,0.2,0.4,0.6,0.8,1},
       scaled ticks=false,
       ytick={0.1,0.05},
       yticklabels={0.1,0.05},
    ]
    \addplot [
        domain=0:1, 
        samples=100, 
        color=black, dashed,
    ]
    {(x * (1 - x)^2)/((2-x)^2)};
    \addlegendentry{\(\kappa = 1\)}
    \addplot [
        domain=0:1, 
        samples=100, 
        color=black,
        ]
     {(x * (1 - x)^2)/((1.5-x)^2)};
    \addlegendentry{\(\kappa = 0.5\)}
    \end{axis}
    \end{tikzpicture}
    \caption{The Capacity of a tandem queue-channel with $m=2$ w.r.t arrival rate $\lambda$.}
    \label{tandemqcap}
\end{figure}

\subsection{Parallel Queue-channels}
In this section, we assume the source transmits the qubits over a quantum network of two parallel queue-channels. Each qubit has a probability $\delta$ and $1- \delta$ of passing through queue 1 and queue 2, respectively. Fig.~\ref{yerasurenetwork} depicts the schematic of this system model. Let $\mu_1$ and $\mu_2$ be the service rates of two queue-channels, respectively.
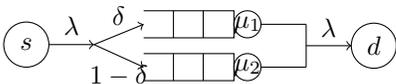
\begin{figure}[h!]
\centering
  \begin{tikzpicture}[scale = 0.6]
  \draw (0.9,-0.25) circle (0.5) node {$s$};
  \draw[->] (1.4,-0.25) -- (2.4,-0.25) node[above,midway]{$\lambda$};
  \draw[->] (2.4,-0.25) -- (3.5,0.2) node[above,midway]{$\delta$};
  \draw[->] (2.4,-0.25) -- (3.5,-0.75) node[below, midway] {$1-\delta$};
  \draw (3.5,0.5) -- (5.5,0.5) -- (5.5,-0.1) -- (3.5,-0.1);
  \draw (4.16,0.5) -- (4.16,-0.1);
  \draw (4.82,0.5) -- (4.82,-0.1);
  \draw (5.8,0.2) circle (0.3) node {$\mu_1$};
  \draw (3.5,-0.45) -- (5.5,-0.45) -- (5.5,-1.05) -- (3.5,-1.05);
  \draw (4.16,-0.45) -- (4.16,-1.05);
  \draw (4.82,-0.45) -- (4.82,-1.05);
  \draw (5.8,-0.75) circle (0.3) node {$\mu_2$};
  \draw (6.1,0.2) -- (7.1,0.2);
  \draw (6.1,-0.75) -- (7.1,-0.75);
  \draw (7.1,0.2) -- (7.1, -0.75);
  \draw[->] (7.1,-0.275) -- (8.1, -0.275) node[above, midway]{$\lambda$};
  \draw (8.6,-0.275) circle (0.5) node {$d$};
  \end{tikzpicture}
\caption{Parallel Erasure Network}
\label{yerasurenetwork}
\end{figure}

 In the following theorem, we provide a classical capacity expression for the network of parallel queue-channels irrespective of the type of the network. Further, we provide a closed-form expression for the classical capacity when the coherence times are exponentially distributed.

\begin{theorem}\label{thm4}
The classical capacity of parallel queue-channels (in bits/sec) with service rates $\mu_1$ and $\mu_2$ is given by
${
    C_p(\lambda) = \lambda (1 - \delta\mathbb{E}_{\pi_1} [ p(W_1)] - (1 - \delta) \mathbb{E}_{\pi_2}[p(W_2)]).
}$
\end{theorem}
\begin{proof}
Note that, given the waiting times of qubits in each queue, the probability of erasure in parallel queue-channels is given as ${\delta \mathbb{E}_{\pi_1}[p(W_1)] + (1-\delta) \mathbb{E}_{\pi_2}[p(W_2)]}$. Hence, following the similar arguments as in \cite[Theorem~1]{prabhaspawc}, we have the desired result.
\end{proof}
\begin{corollary}\label{cor2}
The capacity expression in the above theorem reduces to
\begin{align}\label{opt2}
 C_p(\lambda) = \frac{\lambda \delta (\mu_1 - \lambda \delta)}{\kappa + \mu_1 - \lambda \delta} + \frac{\lambda (1-\delta) (\mu_2 - \lambda + \lambda \delta)}{\kappa + \mu_2 - \lambda + \lambda \delta},
\end{align}
in bits/sec, when the coherence times are exponentially distributed, i.e, when $p(W) = 1 - e^{-\kappa W}$.

\end{corollary}
We now provide an optimal transmission rate and routing probability for the above capacity expression when the network consists of homogeneous servers with service rate $\mu$.

\begin{corollary}\label{cor3}
The optimal transmission rate and optimal routing probability $[\lambda^*, \delta^*]$ that maximizes the total capacity of a homogeneous parallel queue-channels with service rate $\mu$ is given by
${[\lambda^*,\delta^*] = [2(\mu + \kappa - \sqrt{\mu \kappa + \kappa^2}), \frac{1}{2}].}$
\end{corollary}
\begin{proof}
Refer Appendix~\rom{6}-B.
\end{proof}
Further, in the case of heterogeneous servers in the network, we now provide a closed-form expression for the optimal routing probability. Fig.\ref{parallelq} depicts the capacity of a heterogeneous parallel queue-channel with service rates $\mu_1 = 2,\mu_2 = 3,$ and for a fixed transmission rate in the network. We observe that when $\lambda = 1.9$, choosing a queue with service rate $\mu = 2$ with probability $\delta$ beyond the optimal routing probability drastically decreases the capacity. This happens due to the increase in delay-induced decoherence as the transmission rate approaches the service rate of the queue. Similarly, when $\lambda = 1$, we observe a decreasing capacity for $\delta$ beyond the optimal probability. The following proposition characterizes an optimal routing probability $\delta^{*}$ for a given transmission rate in parallel queue-channels.

\begin{proposition}
The optimal routing probability $\delta^*$ for a given ${\lambda \in \min\{\frac{\mu_1}{\delta},\frac{\mu_2}{(1-\delta)}\}}$ that maximizes the total capacity of heterogeneous parallel queue-channels with service rates $\mu_1$ and $\mu_2$ is given by
\begin{align*}
\begin{split}
    \delta^*  &= \frac{\sqrt{\lambda^2 (\kappa + \mu_1)(\kappa + \mu_2) (2\kappa - \lambda+ \mu_1 + \mu_2)^2}}{\lambda^2(\mu_1 - \mu_2)} - \\
    &\hspace{0.1in}\frac{2\kappa^2 \lambda + \lambda^* \mu_1 - 2\lambda \mu_1 \mu_2 + \kappa \lambda (\lambda - 2(\mu_1 + \mu_2))}{\lambda^2(\mu_1 - \mu_2)}.
    \end{split}
\end{align*}
\end{proposition}

\begin{figure}[b]
    \centering
    \begin{tikzpicture}[scale = 0.9]
        \begin{axis}[
        xlabel = \(\delta\),
        ylabel = {\(C_p(\delta),\text{fixed }\lambda \)},
    xmin=0, xmax=1,
    ymin=0, ymax=1.5,
    xtick={0,0.2,0.4,0.6,0.8,1},
    ytick={0.1,0.5,0.9,1.3},
    ]
      \addplot [
        domain=0:1, 
        samples=200, 
        color=black, 
        ]
      {((1.9 *x) * (2-1.9*x)/(3-1.9*x)) + (1.9*(1-x)*(1.1+1.9*x)/(2.1+1.9*x))};
      \addlegendentry{\(\lambda = 1.9\)}
       \addplot [
        domain=0:1, 
        samples=200, 
        color=black, dashed,
    ]
    {(x*(2-x)/(3-x)) + ((1-x)*(2+x)/(3+x))};
    \addlegendentry{\(\lambda = 1\)}
    \end{axis}
    \end{tikzpicture}
    \caption{The Capacity of a heterogeneous parallel queue-channel with service rate $\mu_1=2$ and $\mu_2 = 3$ for a fixed arrival rate $\lambda$ and $\kappa = 1$.}
    \label{parallelq}
\end{figure}
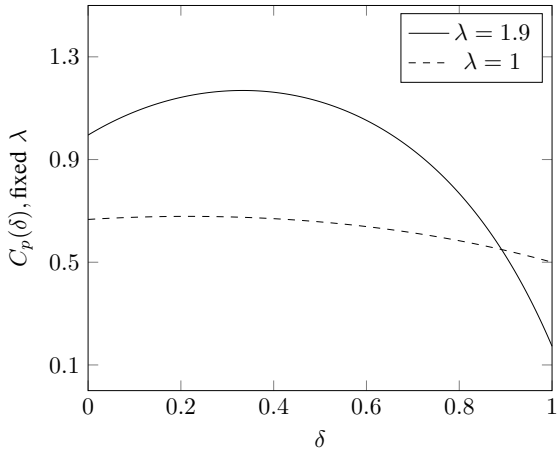

Finally, the next result proves that adding a queue-channel in series with the parallel queue-channels does not change the optimal split in the network.
\begin{corollary}
For a given $\lambda$, addition of queue-channel with service rate $\mu_3$ in series with parallel queue-channels does not affect the optimal routing probability. 
\end{corollary}

Using these tools from Sec.~\ref{sec3}, in the next section, we derive the classical capacity expressions from each source to destination node in a Quantum Jackson Network. 
\section{Classical Capacity of a Quantum Jackson Network}
In this section, we provide our main results on the information capacity from each source to destination node in a general quantum Jackson network setting.

As mentioned earlier in Sec.\ref{sec2}, a quantum Jackson network is a network of $.$/M/1 queue-channels with a single source or multiple sources and a single destination node. Note that we assume each qubit has an i.i.d. Bernoulli probability of traversing to other intermediate nodes after being processed at a node in the network. Recall that $a_{ij}$ denotes the $(i,j)^{th}$ entry of the routing matrix $A,$ and $a_{ij}$ is the probability of routing a qubit from node $i$ to node $j.$ 

Before presenting our main results, we first define a route in a quantum Jackson network.
\begin{definition}
A route $\gamma_s$ from a source node $s$ to the destination node $d$ is a sequence of links between the intermediate nodes that connect a source to the destination, i.e., $\{(s,i_1),(i_1,i_2),\ldots,(i_l,d): i_1, i_2, \ldots i_l \in \mathcal{I}\}$ such that $a_{ij} > 0$ for every link $(i,j) \in \gamma_s$.
\end{definition}

Let $\mathcal{R}^{(s)}$ denote the set of all possible routes from source node `$s$' to the destination node `$d$', and ${\vec{\lambda} = (\lambda_s : s \in S)}$ be the vector of arrival rates at each source node in the network. Define $\mathcal{I}_{\gamma} $ as the sequence of intermediate nodes in route $\gamma$. We then have the following capacity results.

\subsection{Main Results}


\begin{theorem}\label{thm5}
The classical capacity from a source `$s$' to destination `$d$' (in bits/sec) in a repeater-assisted Quantum Jackson Network is given by
$$
   C^{(s)}_{RJ} = \textstyle \sum_{\gamma \in \mathcal{R}^{(s)}} \left(\textstyle \prod_{(i,j) \in \gamma}a_{ij}\right) \zeta_{\gamma},
$$
where ${\zeta_{\gamma} = \mathbb{E}_{\vec{\pi}}\left[ \Pi_{i \in \mathcal{I}_\gamma} q(W_i)\right]}$, $\vec{\pi}$ is the vector of stationary distributions of the waiting times at each node in the network. Further, $W_i$'s are independent and exponential random variables with parameter $\mu_i$ - $\xi_i$, where $\xi_i$ is the net arrival rate at node $i$ in the network satisfying the equations $\xi_i = \sum_{k \in \mathcal{I}} a_{ki} \xi_k , \forall i,k \in \mathcal{I}$.

\end{theorem}
\begin{theorem}\label{thm6}
The classical capacity from a source `$s$' (in bits/sec) to destination `$d$' in a repeater-less Quantum Jackson Network is given by
$$
    C^{(s)}_{LJ} = \textstyle \sum_{\gamma \in \mathcal{R}^{(s)}} \left(\textstyle \prod_{(i,j) \in \gamma}a_{ij}\right) \left(1-\mathbb{E}_{\vec{\pi}}\left[p\left(\sum_{i \in \mathcal{I}_{\gamma}} W_i \right)\right]\right).
$$

\end{theorem}
\begin{proof}
Please refer Appendix~\rom{6}-C for the detailed proofs of Theorem~\ref{thm5} and Theorem~\ref{thm6}.
\end{proof}
The next corollary provides a closed-form expression for the capacity when the coherence times are exponentially distributed. 


\begin{corollary}
The classical capacity from a source `$s$' to destination `$d$' (in bits/sec) in a Quantum Jackson Network when $p(W) = 1-e^{-\kappa W}$ is given by
\begin{align}\label{eqtandem}
    C^{(s)}_J  =\textstyle \lambda_s \sum_{\gamma \in \mathcal{R}^{(s)}} \left( \textstyle \prod_{\{(i,j) \in \gamma, j \in \mathcal{I}\}} a_{ij} \left(\frac{\mu_j - \xi_j}{\kappa + \mu_j - \xi_j}\right)\right).
\end{align}
\end{corollary}

\textit{Remarks:} Note that in a quantum Jackson network, the capacity of a tandem queue-channel over a route $\gamma$ depends on the arrival rates from every source node in the network. This dependency implies that the capacity over a link from a particular source to destination in a Jackson network is indeed affected by the arrivals through other source nodes in the network. Further, we remark that equation~\eqref{eqtandem} holds true only for a \emph{feed-forward} Jackson network, i.e., a qubit does not return to a node after being processed at that particular node.  However, we can also characterize the capacity of a Jackson network with loop-backs, using well-known properties of such networks. We omit the capacity expressions here for the loop-back case in the interest of brevity.

\section{Concluding Remarks}
In this paper, we studied a quantum Jackson network and derived its classical capacity in the presence of waiting time-dependent erasures. We first introduced two simple network topologies, namely tandem queue-channel and parallel queue-channels. We derived the classical capacity expressions, optimal pumping rates and optimal routing probabilities for these simple topologies. Next, using these results as building blocks, we characterized the classical capacity of a general quantum Jackson network. Throughout our work, we dealt with two quantum network models, (i) repeater-assisted and (ii) repeater-less networks. We showed that the two Jackson network models are mathematically equivalent when the coherence times are exponentially distributed. 

For future work, we can utilize the capacity results obtained to derive optimal routing schemes and optimal transmission rates in a quantum Jackson network. Specifically, we could compute the optimal routing scheme for a given transmission rates, as well as the optimal transmission rates for a fixed routing matrix, so as to maximize the capacity. Going further, since the capacity from a source to destination potentially depends on the transmission rates of all sources in the network, there manifests a trade-off between the different source-destination rates. This presents an opportunity to formulate a multi-objective optimization problem with the capacities from each source to destination and solve for a capacity maximizing Pareto frontier for the transmission rates. Finally, since we have modeled only erasures in this paper, studying quantum Jackson networks with more general noise models is a wide-open problem area.

\label{conclusion}


\section{Appendix}
\subsection{Proof of Theorem - \ref{thm1}}\label{ap1}
For simplicity, we prove Theorem~\ref{thm1} for the case of $m=2$ here; however, the same arguments hold for any number of queues in the tandem queue-channel. 

Let $\mathbf{X}$ be the classical bitstream encoded into qubit states $\rho_{\mathbf{X}}$ and transmitted over a repeater-assisted tandem queue-channel. Define $\vec{P} = \{P^n(\mathbf{X})\}_{n=1}^{\infty}$ as the totality of sequences of probability distributions with finite support over $\mathbf{X}$, and $\vec{\boldsymbol{\rho}}$ as the sequences of quantum states corresponding to the encoding $\mathbf{X} \to \rho_{\mathbf{X}}$.

\textit{Upperbound:}
Let $\vec{\mathcal{E}}_{\vec{\mathbf{W}}_k} = \{\mathcal{E}^{(n)}_{\mathbf{W}_k}\}_{n=1}^{\infty}$ be the sequence of channels at repeater $k$ parameterized by the corresponding sojourn times $\{\mathbf{W}_k\}_{n=1}^{\infty}$. Then, following Proposition~1 from \cite{prabhaspawc}, we have the capacity of a repeater-assisted tandem queue-channel denoted by $C_{rt}(\lambda)$ (in bits/sec) as follows:
\begin{align*}
    \begin{split}
        &= \lambda \sup_{\{\vec{P},\vec{\boldsymbol{\rho}}\}} \underline{I}(\{\vec{P},\vec{\boldsymbol{\rho}}\},\vec{\mathcal{E}}_{\vec{\mathbf{W}}_k}:k=1,2) \\
        &\numleq{a} \lambda \sup_{\{\vec{P},\vec{\boldsymbol{\rho}}\}} \liminf\limits_{n \to \infty}\frac{1}{n}\chi(\{P^{(n)},\rho_{\mathbf{X}}\},\mathcal{E}^{(n)}_{\mathbf{W}_k}:k=1,2) \\
        &\numleq{b} \lambda \liminf\limits_{n \to \infty} \frac{1}{n}\sup_{\{\vec{P},\vec{\boldsymbol{\rho}}\}}\chi(\{P^{(n)},\rho_{\mathbf{X}}\},\mathcal{E}^{(n)}_{\mathbf{W}_k}:k=1,2) \\
        &\numeq{c} \lambda \liminf\limits_{n \to \infty} \frac{1}{n} \sum_{i=1}^n \sup_{\{P(X_i),\rho_{X_i}\}} \chi(P(X_i),\mathcal{E}_{W_{2i}}\mathcal{E}_{W_{1i}}(\rho_{X_i})) \\
        &= \lambda \liminf\limits_{n \to \infty} \frac{1}{n} \sum_{i=1}^n [1 - p(W_{1i}) - (1 - p(W_{1i})) p(W_{2i})] \\
        &= \lambda \mathbb{E}_{\vec{\pi}}[1 - p(W_1) - p(W_2) + p(W_1) p(W_2)],
    \end{split}
\end{align*}
where (a) follows from the upper bound derived in \cite[Lemma~5]{gcqc}, (b) holds true from the fact that for each $n$, Holevo information is upperbounded by suppremum of all input encodings, and finally (c) holds true from \cite[Lemma~1]{prabhaspawc}.

\textit{Achievability:}
To prove the achievability, we consider a specific encoding/decoding strategy at the transmitter. Let the classical bits $0$ and $1$ are encoded into fixed orthogonal quantum states $\ket{x_0}$ and $\ket{x_1}$. We assume the decoder at the destination node $d$ measures in a fixed basis. Assuming that the codewords are unentangled across multiple channel uses and the decoder performs the product measurement, we see that the qubits essentially behave like the classical bits, thereby reducing the channel to an induced tandem queue classical channel. Following the similar steps as in \cite[Theorem~1]{qubitspaper}, it can be easily verified that the capacity of the induced tandem queue classical channel is $\lambda \left(1 - \mathbb{E}_{\vec{\pi}}\left[\sum_{i=1}^2 p(W_i) - p(W_1) p(W_2)\right]\right)$ in bit/sec completing the proof. \QEDB

\textit{Remarks:} Note that as a consequence of Burke's Theorem \cite[Theorem~7.6.4]{dsp}, the departure process from each $.$/M/1 queue will be a Poisson process with rate $\lambda$, since the arrival process is Poisson. Consequently, we say each node $i$ in a tandem queue-channel is an M($\lambda$)/M($\mu_i$)/1 queue.

\subsection{Proofs of Corollaries}\label{ap2}
\textit{Proof of Corollary-\ref{corollary2}:} 
Note that in a Jackson network, the waiting times in each M($\lambda$)/M($\mu_i$)/1 node are independent \cite[Theorem~7.6.4]{dsp}, and are exponentially distributed with rate $\mu_i - \lambda$. 

Now, given $p(W) = 1 - e^{-\kappa W},$ the capacity of a tandem queue-channel irrespective of the network type reduces to,
\begin{align*}
\begin{split}
    C_t(\lambda) &= \lambda \textstyle \prod_{i=1}^m \mathbb{E}[e^{-\kappa W_i}] 
    \numeq{a} \lambda \left(\textstyle \prod_{i=1}^m \frac{\mu_i - \lambda}{\kappa + \mu_i - \lambda}\right),
\end{split}  
\end{align*}
where $(a)$ is true due to the fact that $\mathbb{E}[e^{-\kappa W_i}]$ is just the Laplace transform of waiting time $W_i$ evaluated at value $\kappa$.

\textit{Proof of Corollary-\ref{cor2}:}
In a homogeneous parallel queue-channels, it can be seen that for any given $\lambda$, we have equal probability of choosing any queue-channel in the network.
Accordingly, the optimization problem can be modified as
\begin{align}\label{opt3}
\begin{split}
     & \argmax_{\lambda \in [0,2\mu)} \frac{\lambda}{2} (\mathbb{E}_{\pi_1}[e^{-\kappa W_1}] + \mathbb{E}_{\pi_2}[e^{-\kappa W_2}]) \\
      \implies & \argmax_{\lambda \in [0,2\mu)} \lambda \left[\frac{\mu - \lambda/2}{\kappa + \mu - \lambda/2}\right],
\end{split}
\end{align}

where the implication in \eqref{opt3} is due to the fact that sojourn time of each qubit in a M($\lambda$)/M($\mu$)/1 queue-channel is distributed as an independent and exponential random variable of rate $\mu - \lambda$. Now, it can be verified that the optimization problem defined in \eqref{opt3} is a convex optimization problem. Hence, setting the derivative to zero with respect to $\lambda$ gives us the following values for $\lambda$, i.e.,
$$ \lambda = 2(\mu + \kappa \pm \sqrt{\mu \kappa + \kappa^2}).$$
Finally, the only $\lambda$ that lies in the range $[0,2\mu)$ is ${2(\mu + \kappa - \sqrt{\mu \kappa + \kappa^2})}$ completing the proof. \QEDB

Fig.~\ref{parallelqlambda} depicts the plot of capacity of a parallel queue-channel with respect to arrival rate for optimal split $\delta = \frac{1}{2}$.
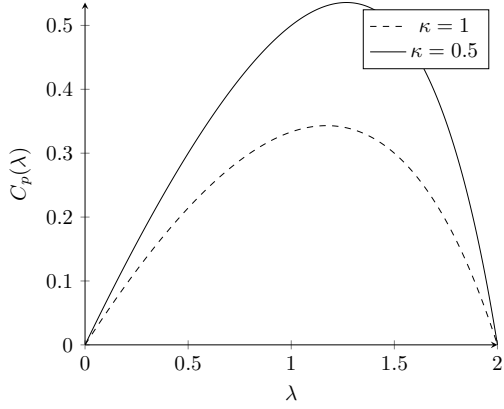
\begin{figure}[t]
    \centering
    \begin{tikzpicture}[scale = 0.8]
        \begin{axis}[
        axis lines = left,
        xlabel = \(\lambda\),
        ylabel = {\(C_p(\lambda) \)},
    ]
    \addplot [
        domain=0:2, 
        samples=200, 
        color=black, dashed,
    ]
    {x*(2-x)/(4-x)};
    \addlegendentry{\(\kappa = 1\)}
    \addplot [
        domain=0:2, 
        samples=200, 
        color=black,
        ]
      {x*(2-x)/(3-x)};
    \addlegendentry{\(\kappa = 0.5\)}
    \end{axis}
    \end{tikzpicture}
    \caption{The Capacity of a  homogeneous parallel queue-channel with service rate $\mu=1$ w.r.t. arrival rate $\lambda$.}
    \label{parallelqlambda}
\end{figure}

\textit{Proof of Corollary-\ref{cor3}:}
Let $C_y(\lambda)$ be the capacity of the network when two parallel queue-channels are connected in series with a queue-channel. Using the capacities derived in Theorem~\ref{thm3} and Theorem~\ref{thm4}, it can be verified that $C_y(\lambda)$ is equivalent to 
\begin{align}\label{opt5}
C_y(\lambda) = C_{y1}(\lambda) + C_{y2}(\lambda),
\end{align}
where $C_{y1}(\lambda)$ and $C_{y2}(\lambda)$ are as follows:
\begin{align*}
    \begin{split}
        C_{y1}(\lambda) &=  \frac{\lambda \delta (\mu_1 - \lambda \delta) (\mu_3 - \lambda)}{(\kappa + \mu_1 - \lambda \delta)(\kappa + \mu_3 - \lambda )},  \\
        C_{y2}(\lambda) &= \frac{\lambda (1 - \delta) (\mu_2 - \lambda (1- \delta)) (\mu_3 - \lambda)}{(\kappa + \mu_2 - \lambda (1-\delta))(\kappa + \mu_3 - \lambda )}.
    \end{split}
\end{align*}
Now given $\lambda$, it can be seen that optimizing $\eqref{opt5}$ with respect to $\delta$ is same as maximizing the capacity in equation \eqref{opt2} with respect to $\delta$. Hence, we have the same optimal split as that of the parallel queue-channels. \QEDB

\subsection{Capacity of a Quantum Jackson Network}\label{ap3}
Let $\mathbf{X}^{(s)}$ be the sequence of classical bit stream encoded into qubit states $\rho_{\mathbf{X}^{(s)}}$ at each source node `$s$' respectively. Define ${\vec{\varrho} = (\rho_{\mathbf{X}^{(s)}} : s \in S)}$ as the sequence of input encodings at all sources in the network. Let ${\vec{\mathcal{W}} = (W_{kj}^{(s)}: s \in \mathcal{S}, k \in \mathcal{I}, j \in \mathbb{Z}^+)}$ be the vector of sojourn times of all qubits from each source at every node $k$ in the network. Accordingly, an erasure network channel in quantum Jackson network can be defined as a map ${\mathcal{N}_{\vec{\mathcal{W}}} : ((S_s(H^I))^{\otimes n}: s \in \mathcal{S}) \to ((S_s(H^O))^{\otimes n}:s \in \mathcal{S})}$ from set of all input Hilbert spaces to output Hilbert spaces respectively. 

Let $M_i \in \mathcal{M}_i$ denote the message to be transmitted from the source node $i$ to the receiver, and $\hat{\mathbf{M}} = (\hat{M}_1, \hat{M}_2, \ldots, \hat{M}_{|\mathcal{S}|})$ denotes the estimated message sequence at the receiver. At every source node $s \in \mathcal{S}$,

\begin{definition}
An $(n, R_s, T_s,\epsilon)$ quantum code consists of the following components:
\begin{itemize}
    \item An encoding function ${\mathbf{X}^{(s)} = f_s(M_s)}$, leading to an encoded n-qubit quantum sequence ${\rho_{\mathbf{X}^{(s)}}}$ corresponding to the message ${M_s}$.
    \item A decoder ${\hat{M}_s = [g(\vec{\Delta}, \mathcal{N}_{\vec{\mathcal{W}}}(\vec{\varrho}, \vec{\mathcal{W}})}]_{s}$, where ${\vec{\Delta} = (\Delta_1, \ldots, \Delta_{|\mathcal{S}|})}$ is the measurement sequence obtained at the receiver. 
\end{itemize}
\end{definition}
Note that $\Delta_i$ is the measurement at the receiver corresponding to the information sent by source $i$, $2^{n \ceil{R_i}}$ is the cardinality of the message set $\mathcal{M}_i$, and $T_i$ is the maximum expected time for all the symbols to the reach the receiver from source $i$.

\begin{definition}
If the decoder chooses $\hat{M}_s$ with average probability of error less than $\epsilon$, then the code is $\epsilon-$achievable. For any $0 < \epsilon < 1$, if there exists an $\epsilon-$achievable code ${(n, R_s,T_s, \epsilon)}$, then the rate ${\hat{R}_s = \frac{R_s}{T_s}}$ is achievable.
\end{definition}

\begin{definition}
The information capacity from a source node $s$ to the destination node $d$ denoted by $C^{(s)}$  in a quantum Jackson network is the supremum of all the achievable rates $\hat{R}_s$ in the network for a given Poisson arrival sequence of rate ${\lambda_s}$.
\end{definition}

Let $\vec{P}_s = \{P^n(\mathbf{X}^{(s)}\}_{n=1}^{\infty}$ as the totality of sequences of probability distributions over input sequences of source $s$, and $\vec{\rho}_s = \{\rho_{\mathbf{X}^{(s)}}\}$ be the sequences of states corresponding to encodings $\mathbf{X}^{(s)} \to \rho_{\mathbf{X}^{(s)}}$. The capacity of a quantum Jackson network from source `$s$' to destination `$d$' is given as:
\begin{align*}
    C^{(s)} = \lambda_s  \sup_{\{\vec{P_s},\vec{\rho_s}\}} \underline{I}(\{\vec{P_s},\vec{\rho_s}\}, \mathcal{N}_{\vec{W}}).
\end{align*}

Using the above definitions, we now prove the capacity of a repeater-assisted quantum Jackson network.

\noindent
\textit{Proof of Theorem~5:}

\textit{Upperbound:} Let $\vec{\mathcal{W}}_s = (W_{kj}^{(s)} : k \in \mathcal{I}, j \in \mathbb{Z}^+)$ be the sequence of sojourn times of all qubits from source $s$ at every node in the network. Define $\mathcal{N}_{\vec{\mathcal{W}}_s} : S_s(\mathcal{H}^I) \to S_s(\mathcal{H}^O)$ as the network channel acting on input sequences of a particular source node $s$. 

Recall, we assume that at any given time instant, decoder is perfectly able to recognize the qubits from each source node. Using the above fact and applying the conditional independence lemma \cite[Lemma~1]{prabhaspawc}, we have the capacity from each source node `$s$' to destination $d$ in a quantum Jackson network as follows:

\begin{align*}
    \begin{split}
        {C}^{(s)}_{RJ} &= \lambda_s  \sup_{\{\vec{P_s},\vec{\rho_s}\}} \underline{I}(\{\vec{P_s},\vec{\rho_s}\}, \mathcal{N}_{\vec{W}_s}) \\
        &\leq   \lambda_s \sup_{\{\vec{P_s},\vec{\rho_s}\}} \liminf\limits_{n \to \infty}\frac{1}{n}\chi(\{P^{(n)},\rho_{\mathbf{X}^{(s)}}\},  \mathcal{N}_{\vec{W}_s}) \\
        &\leq \lambda_s \liminf\limits_{n \to \infty} \frac{1}{n}\sup_{\{\vec{P_s},\vec{\rho_s}\}}  \chi(\{P^{(n)},\rho_{\mathbf{X}^{(s)}}\},  \mathcal{N}_{\vec{W}_s})  \\
        &\numeq{a} \lambda_s \liminf\limits_{n \to \infty} \frac{1}{n} \sum_{i=1}^{n} \sup_{\{P(X_i^{(s)}), \rho_{X_i^{(s)}}\}} \mathcal{X}''(X_i^{(s)})\\
        &=  \textstyle \sum_{\gamma \in \mathcal{R}^{(s)}} \left( \textstyle \prod_{(i,j) \in \gamma} a_{ij}\right) \zeta_{\gamma},
    \end{split}
\end{align*}
where $\mathcal{X}''(X_i^{(s)}) = \sum_{\gamma \in \mathcal{R}_s} \left( \prod_{(i,j) \in \gamma} a_{ij}\right) \mathcal{X}^{'}(X_i^{(s)}) $, $\mathcal{X}'(X_i^{(s)}) =  \mathcal{X}(P(X_i^{(s)}), \mathcal{E}_{\mathbf{W}_{i_l}}\mathcal{E}_{\mathbf{W}_{i_{l-1}}} \ldots \mathcal{E}_{\mathbf{W}_{i_1}}(\rho_{X_i^{(s)}}))$, and ${\zeta_{\gamma} = \mathbb{E}_{\vec{\pi}}\left[ \Pi_{i\in \mathcal{I}_{\gamma}} q(W_i)\right]}$. Further, (a) happens due to the fact that the network channel reduces to the average of tandem queue-channels over all possible routes from source $s$ to destination $d$ (as a consequence of conditional independence lemma). In addition, we know that waiting times in Jackson network at each node $i$ are independent and exponentially distributed with rate $\mu_i - \xi_i$ \cite{frank}, where $\xi_i$ is the net arrival rate at node $i$ in the network satisfying the equations $\xi_i = \sum_{k \in \mathcal{I}} a_{ki} \xi_k , \forall i,k \in \mathcal{I}$ completing the converse part.

\textit{Achievability:} The achievability of the above capacity can be proved by defining independent and orthogonal quantum states, and following the similar steps as in the proof of Theorem~1. \QEDB

\bibliographystyle{IEEEtran}
\bibliography{References}
\balance
  \end{document}